\documentclass[11pt]{article} 

\usepackage[utf8]{inputenc} 

\usepackage{amsmath, amsthm, amsfonts, amssymb, amsrefs}
\usepackage{geometry} 
\usepackage{colonequals}
\usepackage{graphicx} 
\usepackage{booktabs} 
\usepackage{array} 
\usepackage{paralist} 
\usepackage{verbatim} 
\usepackage{subfig} 
\usepackage{fancyhdr} 
\newtheorem{theorem}{\sc Theorem}[section]

\newtheorem{prop}[theorem]{\sc Proposition}
\newtheorem{definition}{Definition}[section]
\newtheorem{example}{Example}[section]
\newtheorem{corollary}[theorem]{\sc Corollary}

\def\C{\mathbb{C}}

\title{Degrees of Freedom and Secrecy in Wireless Relay Networks}
\author{Arsenia Chorti\footnote{Arsenia Chorti is with the Department of Computer Science and Electrical Engineering, University of Essex, Colchester, United Kingdom. email: {achorti@essex.ac.uk}}, Ragnar Freij\footnote{Ragnar Freij is with the Department of Communications and Networking, Aalto University, Espoo, Finland.  email: {ragnar.freij@aalto.fi}}, and David Karpuk\footnote{David Karpuk is with the Department of Mathematics and Systems Analysis, Aalto University, Espoo, Finland.  email: {david.karpuk@aalto.fi}.  The last author is supported by Academy of Finland grant 268364.  This work was partially completed while D.\ Karpuk was a guest of A.\ Chorti at the University of Essex, and thus D.\ Karpuk would like to extend his gratitude to the University of Essex for the excellent hospitality.}}
\date{}
\begin{document}

\maketitle

\begin{abstract}
We translate the problem of designing a secure communications protocol for several users communicating through a relay in a wireless network into understanding certain subvarieties of products of Grassmannians.  We calculate the dimension of these subvarieties and provide various results concerning their defining equations.  When the relay and all of the users have the same number of antennas, this approach places fundamental limits on the amount of data that can be passed through such a network.  
\end{abstract}

\section{Introduction}

Applications of algebra and algebraic geometry to communications engineering traditionally appears in terms of error-correcting codes.  More specifically, error-correcting codes constructed from vector bundles on curves over finite fields are well-known to have desirable properties.  More recently, the advent of wireless communications has opened up new avenues for applied algebraic geometry, which are much less well-tread.  For example, minimizing interference in certain channels can be reduced to cleverly picking subspaces of a given vector space, necessitating the understanding of subvarieties of products of Grassmannians; see \cite{bct}.  In the present article we apply some of these ideas to networks in which users are communicating through a relay.

We will be concerned specifically with transmitting data through a relay, a common scenario in which users who lack a direct link with each other send their information to the relay, who then forwards the information back to the users.  Furthermore, we consider situations in which the relay is untrustworthy, and thus the users would like to pass information through the relay while simultaneously minimizing the amount of information the relay can access.  The relevant wireless communications background will be introduced and explained as is necessary in what follows, and the relevant channel models can be found in, for example, \cite{heath_relay,yener_lattice,ChortiICCC14}.  

Information can be kept secret from a relay by arranging a communications protocol so that the relay only observes certain sums of transmitted data; see \cite{yener_lattice} as well as the articles \cite{ChortiCISS12,ChortiICCC14}.  Effective transmission between two users means minimizing the interference coming from the other users communicating through the same relay.  In algebraic terms, this amounts to pairs of users simultaneously choosing subspaces of a given vector space, subject to the conditions that the secrecy is maintained and the interference is minimized.

In what follows, we translate the problem of designing a secure communications protocol for several users communicating through a relay into understanding certain subvarieties of products of Grassmannians.  We calculate the dimension of these subvarieties and study their defining equations, which give us a hint as to the difficulty of constructing such a protocol.  When the relay and all of the users have the same number of antennas, this approach classifies the available degrees of freedom completely, thus placing fundamental limits on the amount of data that can be passed securely through such a network.  Our approach is largely inspired by \cite{bct}, wherein the authors provide limits on the degrees of freedom in certain interference channels.  


\subsection{Two users}

Suppose we have two users, Alice (user $1$) and Bob (user $2$) who wish to communicate wirelessly with each other, but the only direct link is through a relay, Ray.  The data Alice and Bob wish to send to each other can be represented by complex numbers $x_1,x_2\in\C$, respectively, called \emph{information symbols}.  Alice and Bob transmit this information to Ray, who then forwards it back to both Alice and Bob.  This simple network can be modeled by the following equations (see \cite{heath_relay}):
\begin{align}
r &= h_1x_1 + h_2x_2 + z\\
\tilde{y}_1 &= g_1r = g_1h_1x_1 + g_1h_2x_2 + g_1z + w_1\\
\tilde{y}_2 &= g_2r = g_2h_1x_1 + g_2h_2x_2 + g_2z + w_2
\end{align}
where $h_1, h_2 \in \C$ are the \emph{channel coefficients} which model the effect of the channel from Alice to Ray and Bob to Ray, respectively.  Similarly, $g_1,g_2\in \C$ are the channel coefficients from Ray to Alice and Bob, respectively.  Here, $z$, $w_1$, and $w_2$ represents the additive noise at Ray, Alice, and Bob, which are i.i.d.\ zero-mean circularly symmetric complex Gaussian random variables. 

We assume that knowledge of $h_i$ and $g_i$ is globally available.  In this case, Alice and Bob can subtract off their own data and arrive at the observations
\begin{align}
y_1 &= g_1h_2x_2 + g_1z + w_1 \\
y_2 &= g_2h_1x_1 + g_2z + w_2
\end{align}
The $x_i$ are typically taken from a finite subset $\mathcal{X}\subset \C$ called a \emph{codebook} or \emph{constellation}.  Since the noise is zero mean, Alice and Bob can now decode the other's data by solving
\begin{align}
\hat{x}_2 &= \text{arg}\min_{x\in\mathcal{X}}|g_1h_2x - y_1|^2 \\
\hat{x}_1 &= \text{arg}\min_{x\in\mathcal{X}}|g_2h_1x - y_2|^2
\end{align}
which provides a maximum-likelihood estimate of $x_1$ and $x_2$.  The decoding is successful if $\hat{x}_i = x_i$, and the success rate depends on the variance of the noise.

In general, the relay cannot reliably decode $x_1$ and $x_2$ as it only has access to a linear combination, and essentially would have to attempt to solve one equation in two variables to decode.  If Alice and Bob pre-multiply their symbols $x_i$ by some $u_i\in\C$ such that $h_1u_1 = h_2u_2 = h$, then the observation at the relay will be
\begin{equation}
r = h_1u_1x_1 + h_2u_2x_2 + z = h(x_1 + x_2) + z
\end{equation}
which guarantees that the relay can only observe the sum $x_1 + x_2$.  Typically there will be multiple pairs $(x_1,x_2)$ which give the same sum $x_1 + x_2$.  For example, let us take our constellation to be $\mathcal{X} = \{\pm 1,\pm i\}$.  Suppose that Alice and Bob happen to send $x_1$ and $x_2$ such that $x_1 + x_2 = 0$.   Even if Ray can correctly decode the sum of the symbols to be $0$, he must still guess uniformly at random between the four possible combinations $(1,-1)$, $(-1,1)$, $(i,-i)$, and $(-i,i)$ of transmitted signals.  This general principle can be exploited to achieve non-trivial secrecy in the case of an untrustworthy relay, see \cite{yener_lattice,ChortiICCC14}.

\section{Multi-user communication through a relay}

In the above setup with only two users, each of whom has only one antenna, only one information symbol was exchanged between the two users.  In what follows we add more users and more antennas, both for the users and the relay, and study how this affects the amount of data that can be passed through such a network.  Informally, we are asking the following questions: \emph{What is the maximal $\sum_i d_i$ so that user $i$ can send and receive $d_i$ information symbols, while the relay observes only linear combinations of information symbols?  How does $\sum_i d_i$ depend on the number of relay antennas and the number of user antennas?}

\subsection{The general system model}

We again refer to \cite{heath_relay,tse_book} and the references therein for the equations below which model our channels.  Suppose that we have $K$ users that wish to communicate with each other through a relay, all of whom possibly have more than one antenna.  In general, effective transmission/encoding and reception/decoding of $d$ information symbols at user $i$ requires that user $i$ has $N\geq d$ antennas; essentially each antenna gives a linear equation in the $d$ information symbols, see \cite{tse_book} for more information.  A channel between a user with $N$ antennas and a relay with $M$ antennas is essentially a matrix $H\in \C^{M\times N}$, which we make precise in what follows.

The $i^{th}$ user desires $d_i$ degrees of freedom, that is, they will be transmitting a vector $x_i \in \C^{d_i}$ of information symbols.  We assume that all users have $N$ antennas, that $N\geq d_i$ for all $i$ to allow for linear encoding.  User $i$ uses an $N \times d_i$ encoding matrix $U_i$ to compute a vector $U_ix_i$ of encoded symbols.  We further assume that the relay has $M\geq N$ antennas.

The channel from user $i$ to the relay will be denoted by $H_i\in \C^{M \times N}$, and the channel from the relay to user $k$ by $G_k\in \C^{N\times M}$.  We assume all channel state information is available to all users and all relays, and that $H_i$ and $G_k$ are selected according to some continuous distribution, so that, for example, all channel matrices are generically invertible.

The users transmit their data to the relay simultaneously, so that the relay observes
\begin{equation}
r = \sum_{i = 1}^K H_iU_ix_i + z.
\end{equation}
where $z\in \C^{M\times 1}$ is a column vector of i.i.d.\ zero-mean circularly symmetric complex Gaussian random variables.  The relay then transmits the received message back to the users, so that user $k$ observes
\begin{align}
\tilde{y}_{k} = G_kr = \sum_{i = 1}^KG_kH_iU_ix_i + G_kz + w_k
\end{align}
where $w_k\in \C^{N\times 1}$ is a column vector of i.i.d.\ zero-mean circularly symmetric complex Gaussian random variables.  Since channel information is globally available and user $k$ knows their own data, they can subtract off the quantity $G_kH_kU_kx_k$ and obtain
\begin{equation}
y_k = \sum_{\substack{i = 1 \\ i\neq k}}^K G_kH_iU_ix_i + G_kz + w_k
\end{equation}
from which they attempt to decode some subset of the coordinates of $\{x_1,\ldots,x_K\}$, depending on what data they require.

\subsection{The general strategy}

Let us denote the $M$-dimensional vector space at the relay by $V$.   For each $i = 1,\ldots,K$, we pick for user $i$ subspaces $V_i = \text{colspan}(H_iU_i)$ of $V$, such that $\dim V_i = d_i$ for all $i$.  The space $V_i$ is the target for all of user $i$'s symbols, and the target for any other user sending symbols to user $i$.  So that no symbol arrives in $V_i$ ``unmasked'' and that the symbols arriving from different users are guaranteed to be independent, we insist that
\begin{equation}
V_i = \bigoplus_{\substack{j = 1 \\ j\neq i}}^KV_i \cap V_j
\end{equation}
The relay thus only observes non-trivial linear combinations of information symbols, and thus cannot reliably decode any symbols transmitted to or from user $i$.  

To make full use of all of the $M$ antennas at the relay, we further insist that $V = \sum_i V_i$, so that all of the space at the relay can be occupied by a linear combination of vectors from the users.  To avoid overlapping triples of symbols, we want to further decompose $V$ as
\begin{equation}\label{relay_decompose}
V = \bigoplus_{i = 1}^{K-1}\bigoplus_{j = i + 1}^K V_i \cap V_j
\end{equation}
For user $i$ to decode, we first write
\begin{equation}
V = V_i \oplus I_i,\quad \text{where}\quad  I_i = \bigoplus_{\substack{1\leq j< k\leq K \\ j,k\neq i}}V_j \cap V_k
\end{equation}
is the $i^{th}$ user's interference space, coming from the terms in (\ref{relay_decompose}) which do not involve $i$.  If $G_i$ denotes the channel matrix from the relay to the $i^{th}$ user, then the observation at user $i$ lies in
\begin{equation}\label{sig_int}
G_iV = G_iV_i \oplus G_iI_i
\end{equation}
Since our $G_i$ is assumed to be invertible with probability $1$, we can assume that $\dim(G_iV_i) = d_i$, and hence user $i$ can still linearly decode all of their information symbols.

To decode the symbols contained in $G_iV_i$, user $i$ must first project onto the perp space of $G_iI_i$.  Let $P_i$ denote this projection, and let us, by slight abuse of notation, identify $V_i$ with a matrix whose columns are an orthonormal basis of $V_i$.  The strength of the signal at user $i$ is then dependent on the squared Frobenius norm $||P_iG_iV_i||_F^2$, and the strength of the noise at user $i$ depends on $||P_iG_iz||_F^2 + ||P_iw_i||^2_F$.  The success of user $i$'s decoding is then dependent on the \emph{signal-to-noise ratio}
\begin{equation}
\text{SNR}_i = \frac{||P_iG_iV_i||_F^2}{||P_iG_iz||_F^2 + ||P_iw_i||^2_F}
\end{equation}
Ideally one would like to maximize all of the $\text{SNR}_i$ independently of each other.  However, note that $P_i$ depends on $I_i$ and thus all $V_j$ for $j\neq i$, so one cannot a priori optimize the $\text{SNR}_i$ independently.  Such optimization problems are generally very difficult and depend on the dimension of the space of all strategies, which motivates us to understand the space of all feasible strategies as explicitly as possible.

\begin{example}
Let us demonstrate our strategy with $K = N = M = 3$ and $d_1 = d_2 = d_3 = d = 2$.  Choose a basis $\{v_1,v_2,v_3\}$ of $V$, and set
\begin{equation}
V_1 = \text{span}(v_1,v_2), \quad
V_2 = \text{span}(v_2,v_3), \quad
V_3 = \text{span}(v_1,v_3)
\end{equation}
so that
\begin{equation}
V = (V_1\cap V_3)\oplus (V_1\cap V_2) \oplus (V_2 \cap V_3)
\end{equation}
and pick the encoding matrices $U_i$ so that $(H_iU_i)^{(i)} = v_i$ and $(H_iU_i)^{(i+1)} = v_{i + 1}$, where $A^{(i)}$ means the $i^{th}$ column of the matrix $A$.

Now suppose that $x_i = [x_i^1,x_i^2]^t\in \C^d$ is the vector of information symbols that the $i^{th}$ node wants to send.  Given the above encoding matrices, the observation at the relay can be expressed in the basis $\{v_1,v_2,v_3\}$ as:
\begin{equation}
r = H_1U_1x_1 + H_2U_2x_2 + H_3U_3x_3 + z
= \begin{bmatrix}
x_1^1 + x_3^2 \\
x_1^2 + x_2^1 \\
x_2^2 + x_3^1
\end{bmatrix}
+ z
\end{equation}
Once the relay forwards on this information to the users, we can write user $1$'s observation as
\begin{equation}
\tilde{y}_1 = G_1r = G_1 \begin{bmatrix}
x_1^1 + x_3^2 \\ x_1^2 + x_2^1 \\ 0
\end{bmatrix}
+
G_1 \begin{bmatrix}
0 \\ 0 \\ x_2^2 + x_3^1
\end{bmatrix}
+ g_1z + w_1
\end{equation}
Since user $1$ knows $x_1^1$ and $x_1^2$ as well as the matrix $G_1$, they can subtract off $G_1[x_1^1 \ x_1^2 \ 0]^t$ from $\tilde{y}_1$ and arrive at
\begin{equation}
y_1 = G_1 \begin{bmatrix}
x_3^2 \\ x_2^1 \\ 0
\end{bmatrix}
+
G_1 \begin{bmatrix}
0 \\ 0 \\ x_2^2 + x_3^1
\end{bmatrix}
+ g_1z + w_1
\end{equation}
User $1$ can now recover $x_2^1$ and $x_3^2$ by projecting $y_1$ onto the two-dimensional perp space of $G_1v_3$.  Similarly, user $2$ can recover $x_1^2$ and $x_3^1$, and user $3$ can recover $x_2^2$ and $x_1^1$. \hfill $\blacksquare$
\end{example}

Note that in the above example, we picked not only the subspaces $V_i$ but also specified bases.  This reflects additional flexibility available to the users, who may prefer one basis over another to minimize the probability that the relay can decode their information.  However, this is a subtlety we shall ignore as it does not affect the degrees of freedom.

\section{Feasible strategies for $M = N$}

Throughout this section we assume that the relay has $M = N$ antennas.  We first investigate the resulting parameters defining feasible strategies, and see that this immediately places rather strong conditions on $\sum_i d_i$.  We then study in depth the resulting variety of feasible strategies.

\subsection{Feasible tuples}

Our general strategy in the previous section motivates the following definition and theorem, which tells us for what parameters there exists a viable communications strategy.

\begin{definition}\label{feas}
A tuple $(K,N,d_1,\ldots,d_K)$ is \emph{feasible} if given a $N$-dimensional complex vector space $V$, there exist subspaces $V_1,\ldots,V_K$ satisfying
\begin{itemize}
\item[(i)] $\dim(V_i) = d_i$ for all $i=1,\ldots,K$,
\item[(ii)] $V_i = \bigoplus_{j\neq i} V_i\cap V_j$ for all $i = 1,\ldots, K$, and 
\item[(iii)] $V = \bigoplus_{1\leq i < j \leq K} V_i\cap V_j$
\end{itemize}
If $(K,N,d_1,\ldots,d_K)$ is a feasible tuple, we will call a collection $(V_1,\ldots,V_K)$ of subspaces of $V$ satisfying (i), (ii), and (iii) a \emph{feasible strategy}.
\end{definition}

\begin{theorem}\label{feas_thm_tuple}
The tuple $(K,N,d_1,\ldots,d_K)$ is feasible if and only if
\begin{equation}
\sum_{i = 1}^Kd_i = 2N,
\end{equation}
and $d_i\leq N$ for each $i$.
\end{theorem}
\begin{proof}
Let us first assume that we have a feasible tuple and write $V = \bigoplus_{1\leq i < j\leq K}V_i\cap V_j$.  Let $d_{ij} = \dim(V_i\cap V_j)$.  The assumption $V_i = \bigoplus_{j\neq i}V_i\cap V_j$ implies that $d_i = \dim(V_i) = \sum_{j\neq i} d_{ij}$.  Using this, we have
\begin{align}
\dim(V) & = \sum_{i = 1}^{K-1}\sum_{j = i+1}^K d_{ij} \\
& = \sum_{i=1}^{K-1}\left(d_i - \sum_{j=1}^{i-1} d_{ji}\right) \\ 
& = \sum_{i = 1}^{K-1}d_i - (\dim(V) - \sum_{j=1}^{K-1}d_{jK}) \\
& = \sum_{i = 1}^{K-1}d_i - (\dim(V) - d_K) \\
& = \sum_{i = 1}^K d_i - \dim(V)
\end{align}
Rearranging $\dim(V) = \sum_i d_i - \dim(V)$ gives the result, since $\dim(V) = N$.

To prove the converse, assume that $\sum_i d_i = 2N$. Let $V$ be a complex vector space of dimension $N$ with basis $v_1,\ldots,v_N$, and let $W\supseteq V$ be a complex vector space of dimension $2N$ with basis $v_1,\ldots,v_{2N}$. Let $P:W\to V$ be the projection given by $v_{N+k}\mapsto v_k$ for $k=1,\cdots ,N$.    If we set $D_i=\sum_{j=1}^i d_j$, we can construct the spaces $W_i\subseteq W$ explicitly as $W_i=P\text{span}\left(v_{D_{i-1}+1},\cdots,v_{D_i}\right)$. As $d_i\leq N$, the projection $P$ is injective on each $W_i$, so we can assign $V_i=P(W_i)$. Now, it is clear by construction that $\dim V_i=d_i$, $V_i = \oplus_{j\neq i}V_i\cap V_j$, and that $V=\bigoplus_{1\leq i < j\leq K}V_i\cap V_j$.
\end{proof}

\begin{prop}\label{suff_prop}
Suppose we are given a feasible tuple $(K,N,d_1,\ldots,d_K)$ and subspaces $V_1,\ldots,V_K$ of $V$ such that
\begin{itemize}
\item[(i)] $\dim(V_i) = d_i$ for all $i = 1,\ldots,K$,
\item[(ii)] $V_i = \bigoplus_{j\neq i}V_i\cap V_j$ for all $i = 1,\ldots,K$, and
\item[(iii)] $V_i\cap V_j\cap V_k = 0$ for all distinct triples $(i,j,k)$.
\end{itemize}
Then $(V_1,\ldots,V_K)$ is a feasible strategy.
\end{prop}
\begin{proof}
Let $W = \sum_{1\leq i < j \leq K} V_i\cap V_j \subseteq V$.  By (iii) we have $V_i\cap V_j \cap V_k \cap V_l = 0$ for all quadruples of indices $(i,j,k,l)$, unless $i = k$ and $j = l$.  It follows that $W$ decomposes as a direct sum,
\begin{equation}
W = \bigoplus_{1\leq i < j \leq K}V_i\cap V_j \subseteq V
\end{equation}
from which it follows that $(K,\dim(W),d_1,\ldots,d_K)$ is a feasible tuple.  By Theorem \ref{feas_thm} we must have $\sum_i d_i = 2\dim(W)$ and thus $\dim(W) = N$, hence $W = V$.
\end{proof}

\subsection{The feasible variety}


\begin{definition}\label{feas_var}
Let $(K,N,d_1,\ldots,d_K)$ be a feasible tuple, so that $\sum_i d_i = 2N$, and fix an $N$-dimensional complex vector space $V$.  We define the \emph{feasible variety} to be
\begin{equation}
\mathcal{Z} = \left\{(V_1,\ldots,V_K)\in \prod_{i = 1}^K G(d_i,N)\ |\ (V_1,\ldots,V_K) \text{ is a feasible strategy}\right\}
\end{equation}
where $G(d_i,N)$ denotes the Grassmannian of $d_i$-dimensional subspaces of $V\simeq \C^N$.  
\end{definition}

\begin{example}
Suppose that $K = 4$, $N = 2$, and $d_1 = \cdots = d_4 = 1$.  Constructing a feasible strategy means finding lines $V_1,\ldots,V_4$ in $\C^2$ satisfying (i) and (ii) of Definition \ref{feas}.  This is clearly equivalent to finding two distinct lines $L_1\neq L_2$ in $\C^2$, and then setting, for example $V_1 = V_2 = L_1$ and $V_3 = V_4 = L_2$ (or any other partition of the users into two groups of two).  The feasible variety can thus be identified with $(\mathbb{P}^1(\C)\times \mathbb{P}^1(\C)) - \Delta$, where $\Delta$ is the diagonal.
\end{example}

From now on we restrict to tuples of the form $(K,N,d_1,\ldots,d_K)$ such that $d_1 = \cdots = d_K = d$.  We will abbreviate these as $(K,N,d)$.  In this case $(K,N,d)$ is feasible if and only if $Kd = 2N$.  The following theorem classifies exactly when generic strategies $(V_1,\ldots,V_K)$ are feasible.  From an applications point of view this is useful as it makes clear exactly for which tuples $(K,N,d)$ the users can essentially pick arbitrary subspaces.

\begin{theorem}\label{feas_thm}
Suppose that $(K,N,d)$ is a feasible tuple.  Then $\mathcal{Z}$ is a dense, Zariski open subset of $\prod_{i = 1}^K G(d,N)$ if and only if $(K,N,d)$ is of the form $(2,N,N)$ or $(3,N,2N/3)$.
\end{theorem}
\begin{proof}
Suppose first that we are given a feasible tuple $(K,N,d)$ such that $\mathcal{Z}$ is a dense, Zariski open subset of $\prod_{i = 1}^KG(d,N)$.  Then for generic choice of $(V_1,\ldots,V_K)$, we have
\begin{equation}
d = \dim(V_i) = \dim\bigoplus_{\substack{j = 1 \\ j \neq i}}^K V_i\cap V_j = (K-1)e
\end{equation}
where $e$ is the dimension of the intersection of two generic $d$-dimensional subspaces of $\C^N$.  If $2d \leq N$ then $e = 0$ and hence $d = 0$, thus we may assume $2d > N$.  In this case we have $e = 2d - N$, and hence
\begin{equation}
d = (K-1)(2d-N)
\end{equation}
Substituting $d = 2N/K$ and simplifying yields the equation
\begin{equation}
(K^2 - 5K + 6)N = 0
\end{equation}
and therefore $K = 2$ or $K = 3$.  The relation $Kd = 2N$ now gives the desired tuples.

For the converse, let us start with the tuple $(2,N,N)$.  In this case there is only one possible strategy, namely given by both users selecting $V_1 = V_2 = \C^N$, which is of course feasible.  In the other case of the tuple $(3,N,2N/3)$, it is easy to check that generic choices of $V_1,V_2,V_3$ satisfy the conditions of Proposition \ref{suff_prop}, which completes the proof of the theorem.
\end{proof}

\begin{corollary}
For $K = 3$ users, the feasible variety $\mathcal{Z}$ has dimension $2N^2/3$.
\end{corollary}
\begin{proof}
Since $\mathcal{Z}$ is Zariski open and dense in $\prod_{i = 1}^3 G(2N/3,N)$, it has the same dimension as this variety.  Since $\dim G(d,N) = d(N-d)$, the result follows.
\end{proof}

\subsection{The complement of the feasible variety when $K = 3$}
When $(K,N,d)$ is a feasible tuple, Theorem~\ref{feas_thm} shows that a generic strategy is feasible exactly when $K = 2$ or $K = 3$.  When $K = 2$ the sufficient conditions for $(V_1,V_2)$ to be a feasible strategy given by Proposition \ref{suff_prop} are trivially satisfied. When $K = 3$, this raises the question about the complement of the feasible variety, which is Zariski closed, and hence can be written as the variety defined by an ideal in the function ring on $\prod_{i = 1}^3 G(d,N)$, where $d = 2N/3$.

For a nicely parametrized presentation of this ring, corresponding to a coordinate system on the Grassmannians, we will use the Pl\"ucker embedding
\begin{equation}
\iota: G(d,N)\hookrightarrow \mathbb{P}\left(\bigwedge^{d} \mathbb{C}^N\right).
\end{equation}
This is the natural embedding defined by
\begin{equation}
\iota : \mbox{Span}\{v_1,\cdots, v_d\}\mapsto \C (v_1\wedge\cdots\wedge v_d),
\end{equation} and whose image is the subvariety given by a number of quadratic equations in the $d\times n$ complex variables $\{v_1,\cdots, v_d\}$. These equations are known as the Pl\"ucker relations, and the ideal they generate is denoted $P_d$. 

In the ambient space $\mathbb{P}\left(\bigwedge^{d} \mathbb{C}^N\right)^3$, we can now write down the equations for variety of non-feasible strategies, denoted by $\bar{\mathcal{Z}}$. Indeed, we have
\begin{equation}\begin{split}
\bar{\mathcal{Z}} & = \left\{(V_1,V_2,V_3)\in \prod_{i = 1}^3 G(d,N)\ \big|\ \dim(V_1\cap V_2\cap V_3) > 0\right\}\\
& =\left\{(V_1,V_2,V_3)\in \prod_{i = 1}^3 G(d,N)\ \big|\  \exists w : w| \iota V_1, w| \iota V_2, w | \iota V_3\right\}.
\\& =\left\{(V_1,V_2,V_3)\in \prod_{i = 1}^3 G(d,N)\ \big|\  \exists w : w\otimes w\otimes w | \iota V_1\otimes\iota V_2\otimes \iota V_3\right\}.
\end{split}
\end{equation}
This means that $\iota V_1\otimes\iota V_2\otimes \iota V_3$ lies in the ideal generated by $N$ cubics, indexed by basis vectors of $\C^N$. Note that we have here used the standard identification of the exterior product with its dual. We have also identified the dual of a direct product of spaces with the tensor product of their duals \cite{Tensor_ref}.

We see from the above that $\bar{\mathcal{Z}}$ can be written as a cubic irreducible variety in the product of Grassmannians. The codimension of $\bar{\mathcal{Z}}$ is easily seen to be
\begin{equation}
2N-d-d-d+1 = 1.
\end{equation}

\begin{example}
Suppose that $K = N = 3$ and $d = 2$.  We have the identification $G(2,3)\simeq G(1,3) = \mathbb{P}^2(\mathbb{C})$ given by, for example, fixing an inner product and taking a plane to its perp space.  Under this identification, the feasible variety $\mathcal{Z}$ gets identified with
\begin{equation}
\mathcal{Z} = \{(L_1,L_2,L_3)\in \mathbb{P}^2(\mathbb{C})^3\ |\ \text{span}(L_1,L_2,L_3) = \C^3 \}
\end{equation}
If we give projective coordinates $[x_i:y_i:z_i]$ to the $i^{th}$ copy of $\mathbb{P}^2(\mathbb{C})$, then the complement $\bar{\mathcal{Z}}$ is therefore defined by the single equation
\begin{equation}
\det \left[
\begin{matrix}
x_1 & x_2 & x_3 \\
y_1 & y_2 & y_3 \\
z_1 & z_2 & z_3
\end{matrix}
\right] = 0
\end{equation}
That is, $\bar{\mathcal{Z}}$ is the determinantal variety in $\mathbb{P}^2(\mathbb{C})\times \mathbb{P}^2(\mathbb{C})\times \mathbb{P}^2(\mathbb{C})$, and hence defined by a single irreducible cubic.
\end{example}

\subsection{Dimension of the feasible variety}

Suppose we have fixed some feasible tuple $(K,N,d)$ and some relay space $V$.  To examine further the dimension of the feasible variety, let us further fix some set $\{d_{ij}\}_{1\leq i < j \leq K}$ of desired dimensions for the $V_i\cap V_j$.  In practical terms, this means that we are examining strategies in which users $i$ and $j$ are exchanging $d_{ij}$ information symbols with each other.  The assumption $V_i = \bigoplus_{j\neq i}V_i\cap V_j$ then forces $d = \dim(V_i) = \sum_{j\neq i}d_{ij}$ for all $i$.  Let us denote the corresponding variety of feasible strategies by $\mathcal{Z}^0$.

\begin{theorem}
Suppose that $(K,N,d)$ is a feasible tuple and we fix desired pairwise degrees of freedom $\{d_{ij}\}_{1\leq i < j \leq K}$.  Then the feasible variety $\mathcal{Z}^0$ of strategies is isomorphic to a dense, Zariski open subset of
\begin{equation}
\mathcal{S} = \prod_{i = 1}^{K-1}\prod_{j = i + 1}^K G(d_{ij},N),
\end{equation}
and therefore has dimension
\begin{equation}
\dim(\mathcal{Z}^0) = \dim(\mathcal{S}) = \sum_{i = 1}^{K-1}\sum_{j = i + 1}^K d_{ij}(N-d_{ij})
\end{equation}
\end{theorem}
\begin{proof}
Suppose that $(V_1,\ldots,V_K)$ is a feasible strategy.  We can map it to an element of $\mathcal{S}$ by the obvious map
\begin{equation}
(V_1,\ldots,V_K)\mapsto (V_i\cap V_j)_{1\leq i < j \leq K}
\end{equation}
Conversely, consider the set of all $(V_{ij})_{1\leq i < j \leq K}$ in $\mathcal{S}$ such that
\begin{itemize}
\item[(i)] $V_i := \sum_{j \neq i} V_{ij}$ has dimension $\dim(V_i) = d$ for all $i$, and
\item[(ii)] $V_{ij}\cap V_{kl} = 0$ unless $(i,j) = (k,l)$.
\end{itemize}
One can see easily that (i) and (ii) are generically satisfied, and thus this is a dense, Zariski open subset of $\mathcal{S}$.  Furthermore, the sum in (i) is easily seen to be generically direct, and thus the corresponding tuple $(V_1,\ldots,V_K)$ is a feasible strategy by Proposition \ref{suff_prop}.  The theorem follows easily.
\end{proof}

Let us compute the dimension of $\mathcal{Z}^0$ for some desired pairwise degrees of freedom which are of practical interest.

\begin{example}
Suppose that $d_{ij} = d_{kl}$ for all $(i,j) \neq (k,l)$, so that every pair of users is exchanging the same number of information symbols.  From $V_i = \bigoplus_{j\neq i} V_i\cap V_j$ it follows that $d_{ij} = d/(K-1)$ for all $i,j$.  Using the relation $Kd = 2N$ and the above theorem gives
\begin{align}
\dim(\mathcal{Z}^0) &= \binom{K}{2}d(N-d) \\
& = N^2\left(1 - \frac{2}{K(K-1)}\right)\ \text{whenever $d_{ij} = d/(K-1)$ for all $i,j$}
\end{align}
\end{example}

\begin{example}
Suppose that $K$ is even and we pair up the users so that for $i$ odd, only users $i$ and $i+1$ share their information symbols.  In other words,
\begin{equation}
d_{ij} = \left\{
\begin{array}{ll}
d & \text{$i$ odd and $j = i + 1$} \\
0 & \text{otherwise}
\end{array}
\right.
\end{equation}
Then using the above theorem and the equation $Kd = 2N$ gives
\begin{equation}
\dim(\mathcal{Z}^0) = \sum_{\substack{i = 1 \\ \text{$i$ odd}}}^K d(N-d) = N^2(1-2/K)
\end{equation}
\end{example}

\bibliography{myrefs_new.bib}

\end{document}